\documentclass[a4paper,10pt]{article}
\usepackage{stmaryrd}
\usepackage{amsfonts}
\usepackage{bbm}
\usepackage{amscd}
\usepackage{mathrsfs}
\usepackage{latexsym,amssymb,amsmath,amscd,amscd,amsthm,amsxtra,xypic}
\usepackage[dvips]{graphicx}
\usepackage[utf8]{inputenc}
\usepackage[T1]{fontenc}
\usepackage{lmodern}
\usepackage{amssymb}
\usepackage[all]{xy}
\usepackage{nicefrac,mathtools,enumitem}
\usepackage{microtype}

\newtheorem{thm}{Theorem}[section]
\newtheorem{lem}[thm]{Lemma}
\newtheorem{cor}[thm]{Corollary}
\newtheorem{pro}[thm]{Proposition}
\newtheorem{ex}[thm]{Example}
\newtheorem{rmk}[thm]{Remark}
\newtheorem{defi}[thm]{Definition}

\newcommand{\be }{\begin{equation}}
\newcommand{\ee }{\end{equation}}

\newcommand{\pf}{\noindent{\bf Proof.}\ }

\newcommand{\huaA}{\mathcal{A}}

\newcommand{\huaV}{\mathcal{V}}

\newcommand{\frkg}{\mathfrak g}
\newcommand{\frkh}{\mathfrak h}

\newcommand{\frkk}{\mathfrak k}

\newcommand{\frkX}{\mathfrak X}

\def\qed{\hfill ~\vrule height6pt width6pt depth0pt}

\newcommand{\half}{\frac{1}{2}}

\newcommand{\Id}{\rm{Id}}

\newcommand{\dM}{\mathrm{d}}

\newcommand{\Hom}{\mathrm{Hom}}
\newcommand{\Der}{\mathrm{Der}}

\newcommand{\Aut}{\mathrm{Aut}}

\newcommand{\gl}{\mathfrak {gl}}

\newcommand{\End}{\mathrm{End}}
\newcommand{\ad}{\mathrm{ad}}

\newcommand{\inv}{\mathrm{inv}}

\newcommand{\ve}{\mathrm{v}}
\newcommand{\h}{\mathrm{h}}
\newcommand{\sgn}{\mathrm{sgn}}
\newcommand{\Ksgn}{\mathrm{Ksgn}}

\newcommand{\pat}{\partial_t}
\newcommand{\pau}{\partial_u}
\newcommand{\pae}{\partial_s}
\newcommand{\dc}{\Delta c}
\newcommand{\db}{\Delta b}
\newcommand{\tb}{\widetilde{b}}
\newcommand{\tc}{\widetilde{c}}

\newcommand {\emptycomment}[1]{{}}

\newcommand{\V}{\mathbb{V}}
\begin{document}
\title{
{Integration of Lie 2-algebras and their morphisms
\thanks
 {
The first author is supported by NSFC  (11101179) and SRFDP
(20100061120096).  The second author is supported by the German
Research Foundation (Deutsche Forschungsgemeinschaft (DFG)) through
the Institutional Strategy of the University of G\"ottingen.
 }
} }
\author{Yunhe Sheng  \\
Department of Mathematics, Jilin University,
 Changchun 130012, Jilin, China
\\\vspace{3mm}
email: shengyh@jlu.edu.cn\\
Chenchang Zhu\\
Courant Research Center ``Higher Order Structures'', University of
G$\ddot{\rm{o}}$ttingen\\
email:zhu@uni-math.gwdg.de}
\date{}


\maketitle
\begin{abstract}
Given a strict Lie 2-algebra, we can integrate it
to a strict Lie 2-group by integrating the corresponding Lie algebra crossed
module. On the other hand, the integration procedure of Getzler and
Henriques will also produce a 2-group.  In this paper, we show that
these two integration results are  Morita equivalent. As an
application, we integrate a non-strict morphism between Lie algebra
crossed modules  to a
generalized morphism between their corresponding Lie group crossed modules.
\end{abstract}

\section{Introduction}
Recently people have paid much attention to the integration of
Lie-algebra-like structures, such as that of Lie algebroids
\cite{brahic extension,crainic,tz2}, of $L_\infty$-algebras
\cite{Getzler, henriques} and of Courant algebroids
\cite{LS,MT,shengzhu3}. Here ``integration'' is meant in the same
sense in which a Lie algebra is integrated to a corresponding Lie
group.

For an $L_\infty$-algebra $\frkh$, there is an infinite dimensional
Kan simplicial space $\int \frkh$ constructed in \cite{Getzler,
  henriques}, whose $k$-cells are given by $L_\infty$-algebroid
morphisms $T\Delta^k \to \frkh$. Applying to a strict Lie 2-algebra
$\frkh$, the 2-truncation $\tau_2(\int \frkh)$ is a 2-group which is
believed to be the universal integration of $\frkh$.

On the other hand, a strict Lie 2-algebra (resp. strict Lie 2-group)
one-to-one  corresponds to a Lie algebra (resp. Lie group) crossed
module. Thus, a strict Lie 2-algebra can be easily integrated to a
strict Lie 2-group by integrating its corresponding crossed module.
In this article, we prove that these two integration results are
Morita equivalent. Already noticed in \cite{baez  Lie 2 algebra},
the classifying Postnikov data of Lie 2-algebras is the quotient Lie
algebra in degree 0, a choice of Lie module and of a corresponding
3-cocycle. When the Lie module structure is trivial, as in the case
of string Lie 2-algebra, the above Morita equivalence may be implied
via a homotopy theoretical method (see \cite[Sect.
4.1.3]{ScheiberStasheff}). Our article further provides an explicit
Morita morphism generally regardless the triviality of the Lie
module. We must warn the readers that we treat finite dimensional
case only because we need to use the fact that the second homotopy
group of a Lie group is trivial, and this is true only in the finite
dimensional case.

As $L_\infty$-algebras and their integration play an important role
in higher gauge theory \cite{Baezgauge} and higher Chern-Weil theory
\cite{ScheiberStasheff}, we believe our explicit construction will
have potential application in mathematical physics.

As an application, we use the above result to integrate a nonstrict
morphism between strict Lie 2-algebras to a generalized morphism
between their strict Lie 2-groups. We must mention that an
integration of such morphisms is also provided via the technique of
butterflies in \cite{noohi:morphism}. Here we also provide some
mathematical physics oriented examples of such morphisms: they can
encode 2-term $L_\infty$-modules of $\frkg$ in the sense of
\cite{lada-markl}, or equivalently, 2-term representation up to
homotopy in the sense of \cite{Camilo rep upto homotopy},
non-abelian extensions of $\frkg$, and up to homotopy Poisson
actions of $\frkg$ in the sense of \cite{Severa}. Further
application of the integration  is provided in \cite{shengzhu2}.

{\bf Acknowledgement:} We thank very much the referee for very
helpful comments. Both authors give their  warmest thanks to Courant
Research Center ``Higher Order Structures'', G\"{o}ttingen
University, and  Jilin University, where this work was done during
their visits.

\emptycomment{ Now we turn our attention to apply the integration
method in \cite{crainic, Getzler, henriques, tz2} (path-method) to
morphisms between such generalized Lie structures. We must mention
that an integration of such morphisms is also provided via the
technique of butterflies in \cite{noohi:morphism}.

In this paper, we use the path method to integrate a nonstrict
morphism from a Lie algebra $\frkg$ to a strict Lie 2-algebra
$\frkh_\bullet$. Such morphisms can encode 2-term $L_\infty$-modules
of $\frkg$ in the sense of \cite{lada-markl}, or equivalently,
2-term representation up to homotopy in the sense of \cite{Camilo
rep upto homotopy}, non-abelian extensions of $\frkg$, and up to
homotopy Poisson actions of $\frkg$ in the sense of \cite{Severa}.
Further application of the integration in the first case is provided
in \cite{shengzhu2}.

The spirit of the path method is that a path $g(t)$ in a Lie group
$G$ corresponds to a path $\dot{g}(t) g(t)^{-1}$ in its Lie algebra
$\frkg$. Similarly, the homotopy of paths in $G$ can be translated
to a certain differential equation \eqref{eq:g-homotopy} in $\frkg$,
which is called $\frkg$-homotopy (which is not the usual topological
homotopy of paths in $\frkg$). If $G$ is simply connected, $G$,
being the quotient of the space of $G$-paths by $G$-homotopies, is
totally encoded by the data on $\frkg$. That is, $G$ is the quotient
of the space of $\frkg$-paths by $\frkg$-homotopies. Thus we have
integrated $\frkg$ to $G$. Furthermore, a Lie algebra morphism
$\frkg \xrightarrow{\varphi} \frkh$ can be simply integrated by
looking at the induced map on the path spaces. Then the induced map
descends to the quotients, namely the groups, because it preserves
homotopies. However, when we apply the same method to a (non-strict)
morphism $\frkg \xrightarrow{\phi} \frkh_\bullet$ between Lie
2-algebras, the situation becomes much more complicated because the
induced map on the path space does not preserve homotopies any more.
In the case when $\frkh_\bullet$ is a strict Lie 2-algebra, we find
out the differential equation (Eq. \eqref{eqn:delta b}) controlling
the obstruction of the homotopy. Moreover, we prove that the
homotopy class of the solutions of \eqref{eqn:delta b} only depends
on the homotopy classes of the relevant paths, but not on the paths
themselves. Basing on these facts, we finally construct a Lie
2-group morphism from the path 2-group of $\frkg$ to the simply
connected strict Lie 2-group $H_\bullet$ integrating
$\frkh_\bullet$.}

\section{Equivalence of Integrations}
For an $L_\infty$-algebra $\frkh$, there is an infinite dimensional
Kan simplicial space $\int \frkh$ constructed in \cite{Getzler,
  henriques},
\[ \frkh_k := Hom_{d.g.c.a.}( \wedge^\bullet(\frkh),
\Omega^\bullet(\Delta_k)). \] Here we remind the readers that $\wedge^\bullet(\frkh)$ has a natural
differential graded commutative algebra (d.g.c.a.) structure which generalizes
the Chevalley-Eilenberg complex for a Lie algebra. It is shown in this
paper, that when $\frkh$ is a
Lie algebra, the one-truncation $\tau_1(\frkh)$ is exactly the nerve
of the simply connected Lie group $H$ integrating $\frkh$.

When $\frkh$ is a strict Lie 2-algebra corresponding to the crossed
module $\frkh_1\xrightarrow{\dM} \frkh_0$, the two-truncation
$\tau_2(\frkh)$ is a 2-group. On the other hand, there is another
natural Lie 2-group corresponding to the integrated crossed module
of Lie groups (see Def. \ref{defi:xm-gp})  $H_1 \xrightarrow{\dM}
H_0$, where $H_0$ and $H_1$ are simply connected Lie groups of
$\frkh_0$ and $\frkh_1$ respectively (for this integration see for
example \cite[Remark 3.7]{shengzhu2}).  In this section, we show the
isomorphisms between these two 2-groups.

\subsection{Background on $L_\infty$ algebras}
In this section, we briefly review the notions of
$L_\infty$-algebras and crossed modules of Lie algebras. They both
provide models for strict Lie 2-algebras. \emptycomment{ Then we
concentrate on DGLA morphisms from a Lie algebra to a 2-term DGLA.
We will see that several interesting objects can be described by
such a morphism, including 2-term representations up to homotopy of
Lie algebras, non-abelian extensions of Lie algebras and up to
homotopy Poisson actions.}

$L_\infty$-algebras,  sometimes  called strongly homotopy Lie
algebras,  were introduced by Stasheff \cite{stasheff:shla} as a
model for ``Lie algebras that satisfy Jacobi identity up to all
higher homotopies''. The following convention of $L_\infty$-algebras
is the same as Lada and Markl in \cite{lada-markl}.

\begin{defi}
An $L_\infty$-algebra is a graded  vector space $L=L_0\oplus
L_1\oplus\cdots$ equipped with a system $\{l_k|~1\leq k<\infty\}$ of
linear maps $l_k:\wedge^kL\longrightarrow L$ with degree
$\deg(l_k)=k-2$, where the exterior powers are interpreted in the
graded sense and the following relation with Ksozul sign ``Ksgn'' is
satisfied for all $n\geq0$:
\begin{equation}
\sum_{i+j=n+1}(-1)^{i(j-1)}\sum_{\sigma}\sgn(\sigma)\Ksgn(\sigma)l_j(l_i(x_{\sigma(1)},\cdots,x_{\sigma(i)}),x_{\sigma(i+1)},\cdots,x_{\sigma(n)})=0,
\end{equation}
where the summation is taken over all $(i,n-i)$-unshuffles with
$i\geq1$.
\end{defi}


If $L$ is concentrated in degrees $<n$, we obtain the notion of
 {\bf $n$-term $L_\infty$-algebras}.  A semi-strict Lie 2-algebra can
 be understood as a 2-term $L_\infty$-algebra. a {\bf strict Lie 2-algebra} is  a 2-term $L_\infty$-algebra, in which $l_3$ is zero (see \cite{baez  Lie 2
algebra}).

\begin{defi}
A crossed module of Lie algebras is a quadruple
$(\frkh_1,\frkh_0,dt,\phi)$, where $\frkh_1$ and $\frkh_0$ are Lie
algebras,   $dt:\frkh_1\longrightarrow\frkh_0$ is a Lie algebra
morphism and $\phi:\frkh_0\longrightarrow\Der(\frkh_1)$ is an action
of Lie algebra $\frkh_0$ on Lie algebra $\frkh_1$ as a derivation,
such that
$$
dt(\phi_u(m))=[u,dt(m)]_{\frkh_0},\quad
\phi_{dt(m)}(n)=[m,n]_{\frkh_1}.
$$
Here $\Der(\frkh_1)$ is the derivation Lie algebra of $\frkh_1$ with
the commutation Lie bracket $[\cdot,\cdot]_C$.
\end{defi}

The following result is well known.
\begin{thm}\label{thm:dgla and cm}
There is a one-to-one correspondence between strict Lie 2-algebras
and crossed modules of Lie algebras.
\end{thm}

For the precise relation between the operation $l_2$ and the Lie
brackets $[\cdot,\cdot]_{\frkh_0}$ and  $[\cdot,\cdot]_{\frkh_1}$,
please see \cite{shengzhu2}. The key difference is that
$l_2(m,n)=0$, for any $m,n\in L_1=\frkh_1$, and
$[m,n]_{\frkh_1}=l_2(\dM m,n)\neq0$. On the direct sum
$\frkh_0\oplus \frkh_1$, there is also a Lie bracket
$[\cdot,\cdot]_{\frkh_0\oplus \frkh_1}$, which is the semidirect
product of the Lie algebra $\frkh_0$ and the Lie algebra $\frkh_1$:
\begin{equation}\label{eqn:bracketdirectsum}
  [u+m,v+n]_{\frkh_0\oplus
  \frkh_1}=l_2(u,v)+l_2(u,n)+l_2(m,v)+[m,n]_{\frkh_1}.
\end{equation}

\begin{ex}\label{ep:derivation}
For any Lie algebra $\frkk$, $(\frkk,\Der(\frkk),\ad,\Id)$ is a
crossed module of Lie algebras. We denote by
$\frkk\stackrel{\ad}{\longrightarrow}\Der(\frkk)$ the corresponding
strict Lie 2-algebra.
\end{ex}

\subsection{Background on 2-groups}

A group is a monoid where every element has an inverse. A 2-group is
a monoidal category where every object has a weak inverse and every
morphism has an inverse. Denote the category of smooth Banach
manifolds and smooth maps by $\rm Diff$, a semistrict Lie 2-group is
a 2-group in $\rm DiffCat$, where  $\rm DiffCat$ is the 2-category
consisting of categories, functors, and natural transformations in
$\rm Diff$. In the sequel, all the Lie 2-groups are semistrict.

\begin{defi}
A semistrict Lie 2-group consists of an object $C$ in $\rm DiffCat$
together with
\begin{itemize}
\item[$\bullet$] a {\bf multiplication} morphism (horizontal
multiplication) $\cdot_\h:C\times C\longrightarrow C$,
\item[$\bullet$] identity object $1$,
\item[$\bullet$]an inverse map $\inv:C\longrightarrow C$
\end{itemize}
together with the following natural isomorphisms::
\begin{itemize}
\item[$\bullet$] the {\bf associator}
$$
a_{x,y,z}:(x\cdot_\h y)\cdot_\h z\longrightarrow x\cdot_\h
(y\cdot_\h z),
$$
\item[$\bullet$]the {\bf left} and {\bf right unit}
$$
l_x:1\cdot_\h x\longrightarrow x,\quad r_x:x\cdot_\h
1\longrightarrow x,
$$
\item[$\bullet$]the  {\bf  unit} and {\bf  counit}
$$
i_x:1\longrightarrow x\cdot_\h \inv(x),\quad e_x:\inv(x)\cdot_\h
x\longrightarrow 1,
$$
\end{itemize}
such that the pentagon identity for the associator, the triangle
identity for the left and right unit, the first and second zig-zag
identities  are satisfied. We refer to
\cite[Definition 7.1]{baez 2 group}.
\end{defi}
As pointed out in \cite[Sect. 7]{baez 2 group}, if the category $C$
carries a semistrict Lie 2-group structure, then $C$ must be a Lie
groupoid. We denote the groupoid multiplication in $C$ by
$\cdot_\ve$ (vertical multiplication).

In the special case when $a_{x,y,z},~l_x,~r_x,~i_x,~e_x$ are all
identity isomorphisms, we obtain the concept of a {\bf strict Lie
2-group}. It is well-known that strict Lie 2-groups can be described
by  crossed modules of Lie groups.

\begin{defi}\label{defi:xm-gp}
A crossed module of Lie groups is a quadruple $(H_1,H_0,t,\Phi)$,
where $H_0$ and $H_1$ are Lie groups, $t:H_1\longrightarrow H_0$ is
a Lie group morphism, and $\Phi:H_0\times H_1\longrightarrow H_1 $
is an action of $H_0$ on $H_1$ as automorphisms of $H_1$ such that
$t$ is $H_0$-equivariant:
\begin{equation}\label{cm g 1}
t\Phi_g(h)=gt(h)g^{-1},\quad \forall ~g\in H_0,~h\in H_1,
\end{equation}
and $t$ satisfies the so called Peiffer identity:
\begin{equation}\label{cm g 2}
\Phi_{t(h)}(h^\prime)=hh^\prime h^{-1},\quad\forall ~h,h^\prime\in
H_1.
\end{equation}
\end{defi}

The following result is well-known, see \cite{baez 2 group,Barker}
for more details.
\begin{thm}There is a one-to-one correspondence between crossed
modules of Lie groups and strict Lie 2-groups.
\end{thm}
 Roughly speaking, given a crossed
module $(H_1,H_0,t,\Phi)$ of Lie groups, the corresponding  strict Lie
2-group has $C_0=H_0$ and $C_1= H_0\ltimes H_1$, the
semidirect product of $H_0$ and $H_1$. In this strict Lie 2-group,
the source and target maps $s,~t:C_1\longrightarrow C_0$ are given
by
$$
s(g,h)=g,\quad t(g,h)= t(h)\cdot g,
$$
the vertical multiplication $ \cdot_\ve$ is given by:
\begin{equation}\label{m v}
(g^\prime,h^\prime)\cdot_\ve(g,h) =(g, h^\prime\cdot h),\quad
\mbox{where} \quad g^\prime= t(h)\cdot g,
\end{equation}
the horizontal multiplication $\cdot_\h$ is given by
\begin{equation}\label{m h}
(g,h)\cdot_\h (g^\prime,h^\prime)=(g\cdot g^\prime,h\cdot\Phi_g(
h^\prime)).
\end{equation}

\begin{defi}\label{defi:2 morphism}
Given two Lie 2-groups $C$ and $ C^\prime$, a {\bf unital morphism}
$F:C\longrightarrow C^\prime$ consists of a smooth functor
$(F_0,F_1):C\longrightarrow C^\prime$ equipped with a 2-isomorphism
$$
F_2(x,y):F_0(x)\cdot_\h F_0(y)\longrightarrow F_0(x\cdot_\h y),
$$
such that $F_0(1_C)=1_{C'}$ and the following diagrams commute:
\begin{itemize}
\item[$\bullet$]the compatibility condition of $F_2$ with the associator:
\begin{equation}\label{condition F2}
 \xymatrix{
(F_0(x)\cdot_\h F_0(y))\cdot_\h F_0(z)\ar[r]\ar[d]&F_0(x\cdot_\h y)\cdot_\h F_0(z)\ar[r]&F_0((x\cdot_\h y)\cdot_\h z)\ar[d]\\
F_0(x)\cdot_\h (F_0(y)\cdot_\h F_0(z))\ar[r]&F_0(x)\cdot_\h
F_0(y\cdot_\h z)\ar[r]&F_0(x\cdot_\h (y\cdot_\h z)),}
\end{equation}
\item[$\bullet$]
the compatibility condition of $ F_2$ with the left and right unit:
\begin{equation}
\label{condition F2-lr-unit}\xymatrix{1_{C^\prime}\cdot_\h F_0(x)\ar[r]^{l_{F_0(x)}}\ar[d]^{\Id}& F_0(x) \\
F_0(1_C)\cdot_\h F_0(x)\ar[r]^{F_2(1_C,x)}&F_0(1_C\cdot_\h
x)\ar[u]^{F_1(l_x)}}\quad
\xymatrix{F_0(x)\cdot_\h 1_{C^\prime}\ar[r]^{r_{F_0(x)}}\ar[d]^{\Id}& F_0(x) \\
F_0(x)\cdot_\h F_0(1_C)\ar[r]^{F_2(x,1_C)}&F_0( x\cdot_\h
1_C).\ar[u]^{F_1(r_x)}}
\end{equation}
\end{itemize}
\end{defi}

\emptycomment{

\subsection{Integration (strict cases)}

Given a 2-term DGLA $(\frkh_1\stackrel{\dM}{\longrightarrow}
\frkh_0,l_2)$, the corresponding crossed module of Lie algebras
$(\frkh_1,\frkh_0,dt=\dM,\phi=l_2)$ (see Theorem \ref{thm:dgla and
cm})  easily integrates to a crossed module
of Lie groups $(H_1,H_0,t,\Phi)$, where $H_0$ and $H_1$ are the simply connected Lie
groups integrating $\frkh_0$ and $\frkh_1$ respectively. Therefore, we obtain a strict Lie
2-group  $H_0\ltimes H_1 \Rightarrow H_0$ described in \eqref{m v}
and \eqref{m h} as the
integration of  $(\frkh_1\stackrel{\dM}{\longrightarrow}
\frkh_0,l_2)$. Roughly speaking, we have
 \[\begin{array}{ccccccc}
(\frkh_1,0)& &(\frkh_1,[\cdot,\cdot]_{\frkh_1})& &H_1& & H_0\ltimes H_1\\
\Big\downarrow\vcenter{\rlap{d
}}&\longmapsto&\Big\downarrow\vcenter{\rlap{d
}}&\longmapsto&\Big\downarrow\vcenter{\rlap{$\int \dM$  }}
&\longmapsto&\Big\downarrow\Big\downarrow\vcenter{\rlap{ }}\\
(\frkh_0,[\cdot,\cdot]_{\frkh_0})&
&(\frkh_0,[\cdot,\cdot]_{\frkh_0})& & H_0& & H_0\\
\mbox{DGLA}& &\mbox{crossed module}& &\mbox{crossed module}&
&\mbox{Lie 2-group,}
 \end{array}\]
where $[\cdot, \cdot]_{\frkh_1}$ is defined by
$[m,n]_{\frkh_1}\triangleq l_2(\dM m,n)$.

We can apply the same procedure to a Lie algebra, which is
considered to be a special 2-term DGLA. Furthermore, given  a strict
DGLA morphism from a Lie algebra $\frkg$ to a 2-term DGLA
$\frkh_1\longrightarrow\frkh_0$, which is simply a Lie algebra
morphism $\mu:\frkg\longrightarrow \frkh_0$, we have  a
corresponding  morphism of Lie 2-groups:
\begin{equation}\label{ta:G H0}\begin{array}{ccc}
G&\stackrel{\int \mu}{\longrightarrow} &H_0\ltimes H_1\\
\Big\downarrow\Big\downarrow\vcenter{\rlap{ }}&
&\Big\downarrow\Big\downarrow\vcenter{\rlap{
}}\\
G&\stackrel{\int \mu}{\longrightarrow}  &H_0,
 \end{array}\end{equation}
where $G, H_0$ and $H_1$ are the simply connected Lie groups
integrating $\frkg$, $\frkh_0$ and $\frkh_1$ respectively. We now
explain how to construct this morphism by using the path theory.
Denote by $P(\frkg)$ the $C^1$ path space of $\frkg$ with a
convenient boundary condition,
\begin{equation}\label{eqn:pg}
 P(\frkg)=\{a:
C^1\;\text{morphism from}\;[0,1]\; \text{to}\; \frkg|~
a(0)=a(1)=0,~a^\prime(0)=a^\prime(1)=0\}.
\end{equation}
Then $P(\frkg)$ naturally has a smooth structure of Banach manifold
because we choose $C^1$ paths (see for example
\cite[Sect. 2]{tz2}). From now on, when not specially mentioned, {\em all
morphisms are of $C^1$-classes}.

The concatenation $\odot$ of two paths, $a(t)$ and $b(t)$ in
$P(\frkg)$ is defined as follows:
\begin{equation}\label{composition of paths}
a(t)\odot b(t)=\Big\{\begin{array}{cc}2b(2t)&0\leq t \leq
\half\\
2a(2t)&\half\leq t \leq 1\end{array}.
\end{equation}
The paths $a_0$ and $a_1$ are said to be $\frkg$-{\bf homotopic} and we
write $a_0\thicksim a_1$,  if there exist $C^1$-morphisms $a, b: [0,
1]^{\times 2} \to \frkg$ satisfy the following differential
equation
\begin{equation}\label{eq:g-homotopy}
\pat b(t,s)-\pae a(t,s)=[a(t,s),b(t, s)]_\frkg
\end{equation}
with boundary value $b(0,s)=0$, $b(1,s)=0,$ $a(t, 0)=a_0(t)$ and $a(t,
1)=a_1(t)$. This is  equivalent \cite{brahic-zhu} to the fact that,
$$
a(t,s)dt+b(t,s)ds:TI\times TI\longrightarrow \frkg
$$
is a $C^1$ Lie algebroid morphism and $b(0,s)=b(1,s)=0.$

Then the simply connected Lie group $G$ of $\frkg$ is the quotient
$G\cong P(\frkg)/\sim$ (see \cite[Sect. 1.13]{Lie
groups} for more details). So if $\mu:\frkg\longrightarrow\frkh_0$ is
a Lie algebra morphism, we have
\begin{eqnarray}
\nonumber\pat\mu(b(t,s))-\pae \mu(a(t,s))&=&\mu(\pat b(t,s)-\pae
a(t,s))=\mu[a(t,s),b(t, s)]_\frkg\\
\label{eqn:mu homomorphism}&=&[\mu(a(t,s)),\mu(b(t, s))]_{\frkh_0}.
\end{eqnarray}
Furthermore, it is evident that $\mu(b(0,s))=\mu(b(1,s))=0$, which
implies $\mu(a_i)$ are still homotopic in $\frkh_0$ and therefore
there is a well defined morphism from $G$ to $H_0$, which integrates
$\mu$.}

\subsection{Equivalence of two 2-groups}
Given a Lie algebra $\frkg$, denote by $P\frkg$ the usual $C^2$ path space
in $\frkg$ and  by $P^0\frkg$ the $C^2$ path space in $\frkg$ with a
convenient boundary condition,
\begin{equation}\label{eqn:pg}
 P^0\frkg=\{a:
C^2\;\text{morphism from}\;[0,1]\; \text{to}\; \frkg|~
a(0)=a(1)=0,~a^\prime(0)=a^\prime(1)=0\}.
\end{equation}
Then both $P\frkg$ and $P^0\frkg$ naturally have a smooth structure of Banach manifold
because we choose $C^2$ paths (see for example
\cite[Sect. 2]{tz2}).  From now on, when not specially mentioned, {\em all
morphisms are of $C^2$-classes}.

The paths $a_0$ and $a_1$ are said to be $\frkg$-{\bf homotopic} and we
write $a_0\thicksim a_1$,  if there exist $C^2$-morphisms $a, b: [0,
1]^{\times 2} \to \frkg$ satisfy the following differential
equation
\begin{equation}\label{eq:g-homotopy}
\pat b(t,s)-\pae a(t,s)=[a(t,s),b(t, s)]_\frkg
\end{equation}
with boundary value $b(0,s)=0$, $b(1,s)=0,$ $a(t, 0)=a_0(t)$ and $a(t,
1)=a_1(t)$.  This is  equivalent \cite{brahic-zhu} to the fact that,
$$
a(t,s)dt+b(t,s)ds:TI\times TI\longrightarrow \frkg
$$
is a Lie algebroid morphism and $b(0,s)=b(1,s)=0.$ The
$\frkg$-homotopy also restricts to $P^0\frkg$ (see \cite{tz2}). Then the simply connected Lie group $G$ of $\frkg$ is the quotient
\[G\cong P^0 \frkg/\sim = P \frkg /\sim :=\tau_1(\int \frkg), \] (see \cite[Sect. 1.13]{Lie
groups} for more details).

Next, we recall the construction (in \cite{henriques}) of the
2-group structure of $\tau_2(\int \frkh)$ for a strict Lie 2-algebra
$\frkh$.  Let
\[ P_1 \frkh:= \{a:
C^2\;\text{morphism}~[0,1] \to\frkh \}=\{a: C^2
  \;\text{morphism} ~[0,1]\to \frkh_0 \},  \]
  and
\begin{eqnarray*} P_2  \frkh := \{ (a, b, z): [0,1]^{\times 2} \xrightarrow{(a,
b, z)} \frkh_0 \oplus \frkh_0 \oplus \frkh_1| \partial_t
b-\partial_s a = l_2(a, b)+ \dM z,\\
b(0,s)=b(1,s)=z(0,s)=z(1,s)=0\}.\end{eqnarray*} There is an
equivalence relation $\sim$ defined on $P_2\frkh$: $(a^0, b^0, z^0)
\sim (a^1, b^1, z^1) $ if and only if there are
\[ (a, b, c): [0,1]^{\times 3} \to \frkh_0, \quad (x, y, z):
[0,1]^{\times 3} \to \frkh_1, \]
such that
\begin{eqnarray}
\label{eq:abz} \partial_t b-\partial_s a &= &l_2(a, b) +\dM z, \\
\label{eq:cay} \partial_t c - \partial_u a &= &l_2(a, c) +\dM y, \\
\label{eq:cbx} \partial_s c - \partial_u b &= &l_2(b, c) + \dM x, \\
\label{eq:xyz} \partial_u z - \partial_s y + \partial_t x & =
&l_2(a, x) - l_2(b, y) + l_2(c, z)
\end{eqnarray}
with boundary conditions:
\[ c(t, s, u), x(t, s, u), y(t, s, u)|_{t=0\;\text{ or}\; 1, \;\text{or} \;s=0\;\text{ or}\;
  1} =0, \quad z(t,s,u)|_{t=0\;\text{ or}\; 1} =0,  \]
and
\[ a(t,s,0)=a^0, a(t,s,1)=a^1, b(t,s,0)=b^0, b(t,s,1)=b^1,
z(t,s,0)=z^0, z(t,s,1)=z^1, \]
\[ a(t, 0, u)=a^0(t, 0)=a^1(t,0), \quad a(t, 1, u)=a^0(t, 1)=a^1(t, 1). \]
Then $P_2\frkh/\sim \Rightarrow P_1\frkh$ is a groupoid with the
source and target evaluation of $a$ on $s=0$ and $s=1$ respectively.
Moreover the 2-group of  $\tau_2(\int \frkh)$ is exactly the 2-group
structure on $P_2\frkh/\sim \Rightarrow P_1\frkh$  with vertical
multiplication the concatenation with respect to the parameter $s$
and horizontal multiplication the concatenation with respect to the
parameter $t$. Later on we will give a reparametrized horizontal
multiplication \eqref{eq:atdot} \eqref{eq:btdot} \eqref{eq:ztdot}
for convenience. However, we notice that reparametrization will not
change the class in $P_2 \frkh$:
\begin{lem} Given an element $(a, b, z) \in P_2 \frkh$ and
  reparametrizations $\tau_i: [0, 1] \to [0, 1]$, $(a, b, z) \sim
  (a^\tau, b^\tau, z^\tau) \in P_2\frkh$, where
\[ a^\tau(t, s) = \tau_1'(t) a(\tau_1(t),\tau_2(s)), \quad
b^\tau(t,s)=\tau_2'(s) b(\tau_1(t), \tau_2(s)), \quad
z^\tau(t,s)=\tau'_1(t)\tau'_2(s) z(\tau_1(t), \tau_2(s)).  \]
\end{lem}
\begin{proof}In general, elements in  $P_n\frkh$ are d.g.c.a. morphisms
  $\wedge^\bullet \frkh \to \Omega^\bullet ([0,1]^{\times n})$ with
  certain boundary conditions, and the homotopies $\sim$ are
  d.g.c.a. morphisms  $\wedge^\bullet \frkh \to \Omega^\bullet
  ([0,1]^{\times (n +1)})$ with certain boundary conditions. We define $h:
  [0,1]^{\times (n+1)} \to  [0,1]^{\times n} $
\[h(t_1, \dots, t_{n+1}) :=\big( ((1-t_{n+1})t_1+ t_{n+1} \tau_1(t_1)),
\dots, (1-t_{n+1})t_n+ t_{n+1} \tau_n(t_n) \big), \] pulling back
forms by $h$ provides the desired homotopy. See also \cite[Remark
3.10]{brahic-zhu} and \cite[Lemma 1.5]{crainic} for similar
treatment.
\end{proof}

\emptycomment{
Now given a general DGLA morphism $(\mu, \nu): \frkg \to
(\frkh_0\xrightarrow{\dM} \frkh_1)$, instead of
(\ref{eqn:mu homomorphism}), by \eqref{eqn:DGLA morphism c 1},  we
obtain
\begin{equation}
\pat\mu(b(t,s))-\pae \mu(a(t,s))=[\mu(a(t,s)),\mu(b(t,
s))]_{\frkh_0}+\dM\nu(a(t,s),b(t, s)).
\end{equation}
Therefore, the $\mu(a_i)$'s are no longer homotopic in $\frkh_0$, so
we can not construct a simple morphism as in (\ref{ta:G H0}). We
need some technical lemmas  to solve this problem. In this section,
$H_0, H_1, G$ are the simply connected Lie groups corresponding to
$\frkh_0, \frkh_1$, and $\frkg$ respectively.}

We first construct our equivalence with the help of  a couple of lemmas.

\begin{lem}\label{lem:homotopy}
Let $(a, b, z) \in P_2 \frkh$.
Let $\Delta b:[0,1]^{\times 2}\longrightarrow\frkh_1$ satisfy the following
ordinary differential equation
\begin{equation}\label{eqn:delta b}
\pat\Delta b=l_2(a,\Delta b)-z
\end{equation}
with  initial value $\Delta b(0,s)=0$. Denote
$b+\dM\Delta b $ by  $\widetilde{b}$, then we have
\begin{equation} \label{eq:atb}
\pat\widetilde{b}-\pae a=[a,\widetilde{b}]_{\frkh_0}.
\end{equation}
\end{lem}
\pf The conclusion follows from
\begin{eqnarray*}
\pat\widetilde{b}-\pae a-[a,\widetilde{b}]_{\frkh_0}&=&\pat b +\pat
\dM\Delta b-\pae a-[a,b]_{\frkh_0}-[a,\dM\Delta
b]_{\frkh_0}\\
&=&\pat \dM\Delta b+\dM z-[a,\dM\Delta
b]_{\frkh_0}\\
&=&\dM(\pat\Delta b+z-l_2(a,\Delta b))=0. \qed
\end{eqnarray*}

Since $\Delta b(0,s)=0,$ we have
$\widetilde{b}(0,s)=b(0,s)+\dM\Delta b(0,s)=0$. But
$\widetilde{b}(1,s)=b(1,s)+\dM\Delta b(1,s)=\dM\Delta b(1,s) $ is not necessarily zero and
this is exactly the obstruction of $a(-, 0)$ and $a(-, 1)$ being
homotopic.

\begin{pro}\label{pro:concatenation} With the above notations, the concatenation of $\dM\Delta
b(1,-)$ and $a(-,0)$ is homotopic to $a(-,1)$ in $P \frkh_0$,
i.e. we have
 $$\dM\Delta
b(1,-) \odot a(-,0)\sim a(-,1). $$
where the concatenation $\odot$ of two paths, $a(t)$ and $b(t)$ is defined as follows:
\begin{equation}\label{composition of paths}
a(t)\odot b(t)=\Big\{\begin{array}{cc}2\tau'(t)b(\tau(2t))&0\leq t \leq
\half\\
2\tau'(t)a(\tau(2t))&\half\leq t \leq 1\end{array},
\end{equation}
with a cut-off function $\tau:[0,1] \to [0,1]$ such that
\begin{equation}\label{eq:cutoff}
\tau(0)=0, \quad \tau(1)=1, \quad \tau'(t)>0, \quad \forall t \in
[0,1].
\end{equation}
\end{pro}
\pf Since we have
$
\pat\widetilde{b}-\pae a=[a,\widetilde{b}],
$
there exists a family of paths, $g(t,s)$ in the Lie group $H_0$ such
that
$$
\pat g(t,s)\cdot g(t,s)^{-1}=a(t,s),\quad \pae g(t,s)\cdot
g(t,s)^{-1}=\widetilde{b}(t,s).
$$
Since $\widetilde{b}(0,s)=0$, we know that $g(0,s)$ is fixed. Since
$\widetilde{b}(1,s)\neq 0$, $g(1,s)$ is not a constant path in
$H_0$. So $g(t,0)$ and $g(t,1)$ are not homotopic. However, it is
obvious that the concatenation of $g(1,s)$ and $g(t,0)$ is homotopic
to $g(t,1)$ in the Lie group $H_0$. Therefore, the corresponding
$\frkh_0$-paths are $\frkh_0$-homotopic. \qed

\begin{lem}\label{lem:path action}
Let $\frkh_0$ and $\frkh_1$ be two Lie algebras,
$\phi:\frkh_0\longrightarrow \Der(\frkh_1)$ a Lie algebra morphism,
i.e. Lie algebra $\frkh_0$ acts on Lie algebra $\frkh_1$ as a
derivation. Let $\varphi:H_0\longrightarrow \Aut(\frkh_1)$ be the
Lie group morphism which integrates $\phi$, and $h\in H_0$
represented by $a(t)\in P\frkh_0$, i.e. $h=[a(t)]$. Then for any
$v\in\frkh_1$ and $v(s)\in P\frkh_1$, we have
\begin{itemize}
\item[\rm(1).] $\varphi_h(v)=w(1)$, where $w(t)\in P\frkh_1$ is the
solution of the following ODE:
$$
\frac{d}{dt}w(t)=\phi_{a(t)}w(t)
$$
with the initial value $w(0)=v$.
\item[\rm(2).]$\varphi_hv(s)=w(1,s)$, where $w(t,s)$ is the solution
of the following ODE:
\begin{equation}\label{eqn:wts}
\pat w(t,s)=\phi_{a(t)}w(t,s)
\end{equation}
with the initial value $w(0,s)=v(s)$.
\end{itemize}
Consequently, the corresponding group action of $H_0$ on $H_1$, say
$\Phi:H_0\longrightarrow \Aut(H_1)$ is given by
\begin{equation}\label{eqn:action Phi}
\Phi_{h}([v(s)])=[\varphi_{h}v(s)]=[w(1,s)].
\end{equation}
\end{lem}
\pf Let $h(t):[0,1] \longrightarrow H_0$ be a path such that $h(0)=e,~h(1)=h$ and $
a(t)=\dot{h}(t)h(t)^{-1}. $ Then we have $
\phi_{a(t)}=\dot{\varphi}_{h(t)}\varphi_{h(t)}^{-1},
$
which implies that
$
\phi_{a(t)}\circ\varphi_{h(t)}(v)=\dot{\varphi}_{h(t)}(v).
$
Take $w(t)=\varphi_{h(t)}(v)$, then $w(t)$ satisfies the following
ODE:
$$
\dot{w}(t)=\phi_{a(t)}w(t)
$$
with the initial value $w(0)=\varphi_{h(0)}(v)=\varphi_{e}(v)=v$.
Obviously, $\varphi_h(v)=\varphi_{h(1)}(v)=w(1)$. This completes the
proof of item (1). Item (2) can be proved similarly. \qed\vspace{3mm}

\emptycomment{
Define $P^2(\frkg)$  by
\begin{eqnarray*}
P^2(\frkg)\triangleq\{a(t,s)dt+b(t,s)ds: C^1\;\text{Lie algebroid morphism
  from} \;T(I\times
I)\longrightarrow\frkg| \\\mbox{ such that}\quad
a(0,s)=a(1,s)=b(0,s)=b(1,s)=b(t,1)=b(t,0)=0\}/\sim,
\end{eqnarray*}
where the  equivalence relation $\sim$ on $P^2(\frkg)$ is defined as
follows: $a^0(t,s)dt+b^0(t,s)ds$ and $a^1(t,s)dt+b^1(t,s)ds$ are
said to be equivalent, and we write $a^0dt+b^0ds\thicksim
a^1dt+b^1ds$, if there exists a $C^1$ map $c:I^3\longrightarrow
\frkg$ such that
$$
a(t,s,u)dt+b(t,s,u)ds+c(t,s,u)du:TI^3\longrightarrow \frkg
$$
is a Lie algebroid morphism and
\begin{eqnarray*}
c(t,s,u)|_{t=0,1,s=0,1}&=&0,\\
a(t,s,0)dt+b(t,s,0)ds&=&a^0(t,s)dt+b^0(t,s)ds,\\
a(t,s,1)dt+b(t,s,1)ds&=&a^1(t,s)dt+b^1(t,s)ds.
\end{eqnarray*}

\begin{lem}{\rm \cite[Section 3.3]{z:lie2}}\label{lem:path2group} With the above notations,   $P^2(\frkg)$ is
again a Banach manifold, and $\begin{array}{c}
P^2(\frkg)\\
\downarrow\downarrow\vcenter{\rlap{}}\\P(\frkg)
 \end{array}$ is a semi-strict Lie (Banach) 2-group with nontrivial associator $a_{a_3,a_2,a_1}:(a_3\odot a_2)\odot a_1\longrightarrow a_2\odot(a_2\odot
 a_1)$ and various other 2-morphisms, all given by
 reparametrization, where $P(\frkg)$ is given by \eqref{eqn:pg}.
\end{lem}
}

 For
bigons $\xymatrix@C+2em{
 \bullet &
  \ar@/_1pc/[l]_{a_0}_{}="0"
  \ar@/^1pc/[l]^{a_1}^{}="1"
  \ar@{=>}_z"0";"1"^{}
  \bullet }$ and $\xymatrix@C+2em{\bullet &
  \ar@/_1pc/[l]_{a_0^\dag}_{}="2"
  \ar@/^1pc/[l]^{a_1^\dag}^{}="3"
  \ar@{=>}_{z^\dag}"2";"3"^{}
  \bullet}$, which represent $(a, b, z)
 , (a^\dag,b^\dag, z^\dag) \in P_2 \frkh$ respectively, assume that $\Delta b$
  and $\Delta b^\dag$ are the corresponding solutions of (\ref{eqn:delta
  b}) respectively. We reparametrized the concatenation with respect to
$t$, namely the bigon of horizontal multiplication $\xymatrix@C+2em{
 \bullet &
  \ar@/_1pc/[l]_{a_0}_{}="0"
  \ar@/^1pc/[l]^{a_1}^{}="1"
  \ar@{=>}_z"0";"1"^{}
  \bullet &
  \ar@/_1pc/[l]_{a_0^\dag}_{}="2"
  \ar@/^1pc/[l]^{a_1^\dag}^{}="3"
  \ar@{=>}_{z^\dag}"2";"3"^{}
  \bullet}$ as
\begin{eqnarray} \label{eq:atdot}
a^\ddag(t,s)&=&\left\{\begin{array}{cc} a^\dag(t,2s),& 0\leq t\leq
1,~0\leq
s\leq \half\\
a_1^\dag(t),&0\leq t\leq 1,~ \half\leq
s\leq 1\\
a_0(t-1),& 1\leq t \leq 2,~0\leq s\leq \half\\
a(t-1,2s-1),&1\leq t \leq 2,~\half\leq s\leq 1
\end{array}\right.\\\label{eq:btdot}
b^\ddag(t,s)&=&\left\{\begin{array}{cc} 2b^\dag(t,2s),& 0\leq t\leq
1,~0\leq
s\leq \half\\
0,&0\leq t\leq 1,~ \half\leq
s\leq 1\\
0,& 1\leq t \leq 2,~0\leq s\leq \half\\
2b(t-1,2s-1),&1\leq t \leq 2,~\half\leq s\leq 1,
\end{array}\right.\\ \label{eq:ztdot}
z^\ddag(t,s)&=&\left\{\begin{array}{cc} 2z^\dag(t,2s),& 0\leq t\leq
1,~0\leq
s\leq \half\\
0,&0\leq t\leq 1,~ \half\leq
s\leq 1\\
0,& 1\leq t \leq 2,~0\leq s\leq \half\\
2z(t-1,2s-1),&1\leq t \leq 2,~\half\leq s\leq 1,
\end{array}\right.
\end{eqnarray}
and denote by $\Delta b^\ddag$ the corresponding solution of
(\ref{eqn:delta b}).

\begin{lem}\label{lem:delta b}
With the above notations, we have
\begin{equation}
\Delta b^\ddag(2,s)=\Delta b(1,s)\odot w(1,s),
\end{equation}
where $w(t,s):[0,1]^{\times 2}\longrightarrow\frkh_1$ is the solution of
the following ODE:
\begin{equation}
\pat w(t,s)=l_2(a_0(t),w(t,s))
\end{equation}
with the initial value $w(0,s)=\Delta b^\dag(1,s)$.
\end{lem}
\pf We prove a more general formula:
$$
\Delta b^\ddag(t,s)=\left\{
\begin{array}{cc} 0\odot_s \Delta b^\dag(t,s) ,& 0\leq t\leq
1\\
\Delta b(t-1,s)\odot_s w(t-1,s)&1\leq t \leq 2.
\end{array}\right.
$$
When $0\leq t\leq 1$,
$$
\Delta b^\ddag(t,s)=\left\{
\begin{array}{cc}         2\Delta b^\dag(t,s) ,& 0\leq s\leq
\half\\
0 &\half \leq s \leq 1.
\end{array}\right.
$$
Then, when $0\leq s\leq \half$, $\Delta b^\ddag(t,s)$ satisfies
(\ref{eqn:delta b}) since $\Delta b^\dag(t,s)$ does; when $\half\leq
s \leq 1$, $\Delta b^\ddag(t,s)=0$ obviously satisfies (\ref{eqn:delta b}).

 When $1\leq t \leq 2$, by straightforward
computations, we have
\begin{eqnarray*}
&&\pat\Delta b^\ddag(t,s)\\&=&\pat\Delta b(t-1,s)\odot_s \pat w(t-1,s)\\
&=&\Big(l_2(a(t-1,s),\Delta
b(t-1,s))-z(t-1,s))\Big)\odot_s
l_2(a_0(t-1),w(t-1,s))\\
&=&\left\{\begin{array}{cc} 2l_2(a_0(t-1),w(t-1,2s)),& 0\leq
s\leq
\half,\\
2l_2(a(t-1,2s-1),\Delta
b(t-1,2s-1))-2z(t-1,2s-1),& \half \leq s \leq
1.\end{array}\right.\\
&=&l_2(a^\ddag(t,s),\Delta
b^\ddag(t,s))-z^\ddag(t,s).
\end{eqnarray*}
The last equality holds because when $1\leq t\leq 2,~0\leq s\leq
\half$, $b^\ddag(t,s)=0$. \qed\vspace{3mm}

Define $\Psi_1:P_2\frkh{\longrightarrow} H_0\ltimes H_1$ by
\begin{equation}\label{Psi1}
\Psi_1\Big( \xymatrix@C+2em{
  \bullet &
  \ar@/_1pc/[l]_{a_0}_{}="0"
  \ar@/^1pc/[l]^{a_1}^{}="1"
  \ar@{=>}_z"0";"1"^{}
  \bullet}\Big)=\big([a_0],[\Delta
b(1,s)]\big),
\end{equation}
in which $\Delta b$ is the unique solution of (\ref{eqn:delta b})
with the initial value $\Delta b(0,s)=0$. To see that $\Psi_1$ is well
defined, for two elements $(a^0, b^0, z^0)$ and $(a^1,
b^1, z^1)$ in $P_2\frkh$ equivalent through $(a,b,c,x,y,z)$,  \emptycomment{ $A^0=a^0(t,s)dt+b^0(t,s)ds$ and
$A^1=a^1(t,s)dt+b^1(t,s)ds$ between them, i.e. there is a Lie
algebroid morphism:
$$
\huaA:a(t,s,u)dt+b(t,s,u)ds+c(t,s,u)du:TI^3\longrightarrow \frkg,
$$
satisfying $c(t,s,u)|_{t=0,1,s=0,1}=0$ such that
\begin{eqnarray*}
a(t,s,0)dt+b(t,s,0)ds&=&A^0,\qquad
a(t,0,u)=a_0,\\a(t,s,1)dt+b(t,s,1)ds&=&A^1,\qquad a(t,1,u)=a_1,
\end{eqnarray*}}
\[
\xymatrix@C+2em{
  \bullet &
  \ar@/_1pc/[l]_{a_0}_{}="0"
  \ar@/^1pc/[l]^{a_1}^{}="1"
  \ar@{=>}"0";"1"^{z^0}
  \bullet,
}\quad \xymatrix@C+2em{
  \bullet &
  \ar@/_1pc/[l]_{a_0}_{}="0"
  \ar@/^1pc/[l]^{a_1}^{}="1"
  \ar@{=>}"0";"1"^{z^1}
  \bullet},
  \quad \huaA=\xymatrix@C+2em{
  \bullet &
  \ar@/_1pc/[l]_{a_0}_{}="0"
  \ar@/^1pc/[l]^{a_1}^{}="1"
  \ar@{=>}@/^1pc/"0";"1"^{z^1}="2"
  \ar@{=>}@/_1pc/"0";"1"^{z^0}="3"
  \bullet,
}
\]
we need to prove that $\Delta b(1,s,0)$ and $\Delta b(1,s,1)$ are
homotopic in the Lie algebra $\frkh_1$. This follows from the next
lemma.
\begin{lem} Let $(a, b, c, x, y, z)$ be as above.
Let $\Delta b$, $\Delta c$: $[0,1]^{\times 3} \to \frkh_1$ be the solution
of the following ordinary
differential equations \begin{eqnarray} \label{eqn:delta b 1}
\pat\Delta b&=&l_2(a,\Delta b)-z,\\
\pat\Delta c&=&l_2(a,\Delta c)-y,
\end{eqnarray}
with the initial value $\Delta b(0,s,u)=\Delta c(0,s,u)=0$, we have
\begin{equation}\label{eqn:db dc}
\pae\Delta c(1,s,u)-\pau\Delta b(1,s,u)=[\Delta b(1,s,u),\Delta
c(1,s,u)]_{\frkh_1}
\end{equation}
and
\begin{equation}\label{eqn:dc}
\Delta c(1,1,u)=0.
\end{equation}
Hence $\Delta b(1,s,0)$ and $\Delta b(1,s,1)$ are homotopic.
\end{lem}
\pf Denote $\widetilde{b}=b+\dM\Delta
b,~\widetilde{c}=c+\dM\Delta c$, by Lemma \ref{lem:homotopy},
we have
\begin{eqnarray}
\label{eqn:b a}\pat\widetilde{b}-\pae a &=&[a,\widetilde{b}]_{\frkh_0},\\
\label{eqn:c
a}\pat\widetilde{c}-\pau a&=&[a,\widetilde{c}]_{\frkh_0}.
\end{eqnarray}

 Denote $\overline{a}=a$, $\overline{b}=b+\Delta b$ and
$\overline{c}=c+\Delta c$. By (\ref{eqn:delta b 1}) and
(\ref{eqn:DGLA morphism c 1}), we have\footnote{See
\eqref{eqn:bracketdirectsum} for the definition
$[\cdot,\cdot]_{\frkh_0\oplus \frkh_1}$}
\begin{eqnarray*}
\pat\overline{b}-\pae\overline{a}
&=&[a,b]_{\frkh_0}+\dM z +l_2(a,\Delta
b)-z\\
&=&[\overline{a},\overline{b}]_{\frkh_0\oplus \frkh_1}+\dM z-z.
\end{eqnarray*}
Similarly, we have
\begin{eqnarray}\label{eq:ca}
\pat\overline{c}-\pau\overline{a}=
[\overline{a},\overline{c}]_{\frkh_0\oplus \frkh_1} +\dM y- y.
\end{eqnarray}
Next we prove that
\begin{equation} \label{eq:bcbb}
\pae\overline{c}-\pau\overline{b}=
[\overline{b},\overline{c}]_{\frkh_0\oplus \frkh_1}+\dM x -x.
\end{equation}

 By straightforward computations, we have
\begin{eqnarray}
\nonumber&&\pat\big(\pae\overline{c}-\pau\overline{b}-
[\overline{b},\overline{c}]_{\frkh_0\oplus \frkh_1}
\big)\\&=&\pae\pat\overline{c}-\pau\pat\overline{b}-[\pat\overline{b},\overline{c}]_{\frkh_0\oplus
\frkh_1}
-[\overline{b},\pat\overline{c}]_{\frkh_0\oplus \frkh_1}\\
\nonumber&=&\pae\big(\pau\overline{a}+
[\overline{a},\overline{c}]_{\frkh_0\oplus \frkh_1} +\dM y -y \big)-[\pat\overline{b},\overline{c}]_{\frkh_0\oplus\frkh_1}\\
\nonumber&&-\pau\big(\pae\overline{a}+
[\overline{a},\overline{b}]_{\frkh_0\oplus\frkh_1}+\dM z-z\big)-[\overline{b},\pat\overline{c}]_{\frkh_0\oplus\frkh_1}\\
\nonumber&=&[\pae\overline{a},\overline{c}]_{\frkh_0\oplus\frkh_1}-[\pat\overline{b},\overline{c}]_{\frkh_0\oplus\frkh_1}
+[\overline{a},\pae\overline{c}]_{\frkh_0\oplus\frkh_1}+\pae(\dM y - y)\\
\nonumber&&-[\pau\overline{a},\overline{b}]_{\frkh_0\oplus\frkh_1}-[\overline{b},\pat\overline{c}]_{\frkh_0\oplus\frkh_1}
-[\overline{a},\pau\overline{b}]_{\frkh_0\oplus\frkh_1}-\pau(\dM z
- z)\\
\nonumber&=&[\overline{c},\dM z
-z]_{\frkh_0\oplus\frkh_1}-[\overline{b},\dM y-y]_{\frkh_0\oplus\frkh_1}
+[\overline{a},\pae\overline{c}-\pau\overline{b}-
[\overline{b},\overline{c}]_{\frkh_0\oplus\frkh_1}]_{\frkh_0\oplus\frkh_1}\\
&&\label{eqn:t4}+\pae(\dM y - y)-\pau(\dM z- z).
\end{eqnarray}
Meanwhile, we have
\begin{eqnarray}
\nonumber[\overline{c},\dM z -z ]_{\frkh_0\oplus\frkh_1}&=&[c+\Delta
c,\dM z-z]_{\frkh_0\oplus\frkh_1}\\\label{eqn:t1}
&=&\dM l_2(c,z)-l_2(c,z),\\\nonumber
~[\overline{b},\dM y- y]_{\frkh_0\oplus\frkh_1}&=&[b+\Delta
b,\dM y-y]_{\frkh_0\oplus\frkh_1}
\\\label{eqn:t2}&=&\dM l_2( b,y)-l_2(b,y),
\end{eqnarray}
\emptycomment{
and
\begin{eqnarray}
\nonumber&&\pae(\dM\nu(a,c)-\nu(a,c))-\pau(\dM\nu(a,b)-\nu(a,b))\\
\nonumber&&=\dM\nu(\pae a,c)+\dM\nu( a,\pae c)-\nu(\pae a,c)-\nu(
a,\pae c)\\
\nonumber&&-\dM\nu(\pau a,b)-\dM\nu( a,\pau b)+\nu(\pau a,b)\nu(
a,\pau
b)\\
\nonumber&&=\dM\nu(\pat
b-[a,b]_\frkg,c)+\dM\nu(a,[b,c]_\frkg)-\dM\nu(\pat
c-[a,c]_\frkg,b)\\
\nonumber&&-\nu(\pat b-[a,b]_\frkg,c)-\nu(a,[b,c]_\frkg)+\nu(\pat c-[a,c]_\frkg,b)\\
&&\label{eqn:t3}=\dM\nu(a,[b,c]_\frkg)+c.p.+\nu([a,b]_\frkg,c)+c.p.+\pat(\dM\nu(b,c)-\nu(b,c)).
\end{eqnarray} }
By (\ref{eqn:t1}), (\ref{eqn:t2}),  (\ref{eqn:t4}) and
\eqref{eq:xyz},
 we obtain that
\begin{equation}\label{eqn:t5}
\pat\big(\pae\overline{c}-\pau\overline{b}-
[\overline{b},\overline{c}]_{\frkh_0\oplus\frkh_1}-\dM x+x\big)=[\overline{a},\pae\overline{c}-\pau\overline{b}-
[\overline{b},\overline{c}]_{\frkh_0\oplus\frkh_1}-\dM x+x]_{\frkh_0\oplus\frkh_1}.
\end{equation}
Since we have
$$
\big(\pae\overline{c}-\pau\overline{b}-
[\overline{b},\overline{c}]_{\frkh_0\oplus\frkh_1}-\dM x+ x\big)|_{t=0}=0,
$$
by  the uniqueness of solutions of ordinary differential equations,
(\ref{eqn:t5}) implies \eqref{eq:bcbb}.

Since $b(1,s,u)=c(1,s,u)=0$, we have
$$
\pae\Delta c(1,s,u)-\pau\Delta b(1,s,u)=[\Delta b(1,s,u),\Delta
c(1,s,u)]_{\frkh_1}.
$$
By the initial value condition  $a(t,1,u)=a(t,1,0),$ for any $u$,
since $\bar{a}=a$, we
have $\pau \bar{a}(t,1,u)=0$. By (\ref{eq:ca}) and $c(t,1,u)=0$, we have
$$
\pat \overline{c}(t,1,u)=[\overline{a}(t,1,u),
\overline{c}(t,1,u)]_{\frkh_0\oplus\frkh_1}.
$$
Since $\overline{c}(0,1,u)=0$, it follows that
$\overline{c}(t,1,u)=0$, which implies that $\Delta c(t,1,u)=0$ and
thus $\Delta c(1,1,u)=0$. Therefore, $\Delta b(1,s,0)$ and
$\Delta b(1,s,1)$ are homotopic. \qed\vspace{3mm}

Morita equivalence are defined for $n$-groupoids in an arbitrary
category with a certain Grothendieck pretopology in \cite{z:tgpd2}.
We adapt this notation to our situation: a morphism $F: C \to C'$ of
Lie 2-group is a {\bf hypercover} and denoted by  $F: C
\xrightarrow{\sim} C'$, if $F_0: C_0\to C'_0$ is a surjective
submersion, and the natural map
\[ C_1 \to C'_1 \times_{ C'_0\times C'_0} C_0 \times C_0 \]
is an isomorphism. Moreover $C$ and $C'$ are {\bf Morita equivalent},
if there is another Lie 2-group $C''$ such that there are hypercovers
 $C\xleftarrow{\sim} C''\xrightarrow{\sim} C'$. A {\bf
generalized morphism} between Lie 2-groups is a span of morphisms
$C\xleftarrow{\sim} C''\xrightarrow{} C'$.

With the above preparations, we have
\begin{thm}\label{thm:main 3} There is a Lie 2-group Morita
  equivalence given by a morphism $(\Psi_0, \Psi_1, \Psi_2=id)$:
\begin{equation}\label{main map}\begin{array}{ccc}
P_2 \frkh /\sim &\stackrel{ \Psi_1}{\longrightarrow} &H_0\ltimes H_1\\
\Big\downarrow\Big\downarrow\vcenter{\rlap{ }}&
&\Big\downarrow\Big\downarrow\vcenter{\rlap{
}}\\
P_1\frkh&\stackrel{\Psi_0}{\longrightarrow}  &H_0,
 \end{array}\end{equation}
where $\Psi_0(a(t))=[a(t)]$, which is the equivalence class of the
path $a(t)$ and $\Psi_1$ is given by \eqref{Psi1}.
\end{thm}
\begin{rmk} There is an integration obstruction proved in
  \cite{henriques}, that is, the quotient $P_2 \frkh /\sim$ might not
  be representable as a Banach manifold unless a certain obstruction
  class vanishes. In this theorem, we show directly (Prop. \ref{prop:iso})
  that
  $\tau_2(\int \frkh)$ is always representable.
\end{rmk}

We prove it by several steps.

\begin{lem} The above morphism $(\Psi_0, \Psi_1, \Psi_2)$ is a 2-group morphism.
\end{lem}
\pf
 Obviously, $(\Psi_1,\Psi_0)$ respects the source and target
maps. It is not hard to see  that $(\Psi_1,\Psi_0)$ is a morphism
with respect to the vertical multiplication. In fact, for $
\xymatrix@C+2em{
  \bullet &
  \ar@/_1pc/[l]_{a_0}_{}="0"
  \ar@/^1pc/[l]^{a_1}^{}="1"
  \ar@{=>}"0";"1"^{}
  \bullet}$, $
\xymatrix@C+2em{
  \bullet &
  \ar@/_1pc/[l]_{a_1}_{}="0"
  \ar@/^1pc/[l]^{a_2}^{}="1"
  \ar@{=>}"0";"1"^{}
  \bullet}\in P^2\frkh$, assume that $\Delta b, ~\Delta b^\sharp$ are the corresponding solutions of (\ref{eqn:delta b}) respectively. By definition, we have
$$
\Psi_1(\xymatrix@C+2em{
  \bullet &
  \ar@/_1pc/[l]_{a_1}_{}="0"
  \ar@/^1pc/[l]^{a_2}^{}="1"
  \ar@{=>}"0";"1"^{}
  \bullet})\cdot_\ve\Psi_1(\xymatrix@C+2em{
  \bullet &
  \ar@/_1pc/[l]_{a_0}_{}="0"
  \ar@/^1pc/[l]^{a_1}^{}="1"
  \ar@{=>}"0";"1"^{}
  \bullet})=([a_0],[\Delta b^\sharp(1,s)]\cdot
[\Delta b(1,s)]).
$$
On the other hand, it is straightforward to see that $\Delta
b^\sharp\odot \Delta b$ is the solution of (\ref{eqn:delta b}) for
the bigon $\xymatrix@C+2em{
   \bullet & \bullet
    \ar@/_1.5pc/[l]_-{a_0}^{}="0"
    \ar[l]_{a_1\qquad}_{}="1"
    \ar@/^1.5pc/[l]^-{a_2}_{}="2"
    \ar@{=>} "0";"1"^{}
    \ar@{=>} "1";"2"^{}
  }$. Therefore, we have
  $
\Psi_1(\xymatrix@C+2em{
   \bullet & \bullet
    \ar@/_1.5pc/[l]_-{a_0}^{}="0"
    \ar[l]_{a_1\qquad}_{}="1"
    \ar@/^1.5pc/[l]^-{a_2}_{}="2"
    \ar@{=>} "0";"1"^{}
    \ar@{=>} "1";"2"^{}
  })=([a_0],[\Delta b^\sharp(1,s)\odot
\Delta b(1,s)]),
  $
which implies that
\begin{eqnarray*}
\Psi_1(\xymatrix@C+2em{
   \bullet & \bullet
    \ar@/_1.5pc/[l]_-{a_0}^{}="0"
    \ar[l]_{a_1\qquad}_{}="1"
    \ar@/^1.5pc/[l]^-{a_2}_{}="2"
    \ar@{=>} "0";"1"^{}
    \ar@{=>} "1";"2"^{}
  })&=&\Psi_1(\xymatrix@C+2em{
  \bullet &
  \ar@/_1pc/[l]_{a_1}_{}="0"
  \ar@/^1pc/[l]^{a_2}^{}="1"
  \ar@{=>}"0";"1"^{}
  \bullet})\cdot_\ve\Psi_1(\xymatrix@C+2em{
  \bullet &
  \ar@/_1pc/[l]_{a_0}_{}="0"
  \ar@/^1pc/[l]^{a_1}^{}="1"
  \ar@{=>}"0";"1"^{}
  \bullet}).\end{eqnarray*}

Next we prove that $(\Psi_1,\Psi_0)$ is also a morphism with respect
to the horizontal multiplication.  By (\ref{m h}), we have
\begin{eqnarray*}
\Psi_1(\xymatrix@C+2em{
 \bullet &
  \ar@/_1pc/[l]_{a_0}_{}="0"
  \ar@/^1pc/[l]^{a_1}^{}="1"
  \ar@{=>}"0";"1"^{}
  \bullet })\cdot_\h\Psi_1( \xymatrix@C+2em{\bullet &
  \ar@/_1pc/[l]_{a_0^\dag}_{}="2"
  \ar@/^1pc/[l]^{a_1^\dag}^{}="3"
  \ar@{=>}"2";"3"^{}
  \bullet})&=&([a_0],[\Delta b(1,s)])\cdot_\h([a_0^\dag],[\Delta b^\dag(1,s)])\\
&=&\Big([a_0\odot a_0^\dag],[\Delta
b(1,s)]\cdot\Phi_{[a_0]}\big([\Delta b^\dag(1,s)]\big)\Big),
\end{eqnarray*}
where $\Phi$ is given by (\ref{eqn:action Phi}) which integrates the
action of $\frkh_0$ on $\frkh_1.$

 On the other hand, by Lemma \ref{lem:delta b} and Lemma
\ref{lem:path action}, we have
\begin{eqnarray*}
\Psi_1(\xymatrix@C+2em{
 \bullet &
  \ar@/_1pc/[l]_{a_0}_{}="0"
  \ar@/^1pc/[l]^{a_1}^{}="1"
  \ar@{=>}"0";"1"^{}
  \bullet &
  \ar@/_1pc/[l]_{a_0^\dag}_{}="2"
  \ar@/^1pc/[l]^{a_1^\dag}^{}="3"
  \ar@{=>}"2";"3"^{}
  \bullet}
)&=&\Big([a_0\odot a_0^\dag],[\Delta
b^\ddag(2,s)]\Big)\\&=&\Big([a_0\odot a_0^\dag],[\Delta b(1,s)\odot
w(1,s)]\Big)\\&=&\Big([a_0\odot a_0^\dag],[\Delta
b(1,s)]\cdot\Phi_{[a_0]}\big([\Delta b^\dag(1,s)]\big)\Big),
\end{eqnarray*}
 which implies that
\begin{eqnarray*}
\Psi_1(\xymatrix@C+2em{
 \bullet &
  \ar@/_1pc/[l]_{a_0}_{}="0"
  \ar@/^1pc/[l]^{a_1}^{}="1"
  \ar@{=>}"0";"1"^{}
  \bullet &
  \ar@/_1pc/[l]_{a_0^\dag}_{}="2"
  \ar@/^1pc/[l]^{a_1^\dag}^{}="3"
  \ar@{=>}"2";"3"^{}
  \bullet}
)=\Psi_1(\xymatrix@C+2em{
 \bullet &
  \ar@/_1pc/[l]_{a_0}_{}="0"
  \ar@/^1pc/[l]^{a_1}^{}="1"
  \ar@{=>}"0";"1"^{}
  \bullet })\cdot_\h\Psi_1( \xymatrix@C+2em{\bullet &
  \ar@/_1pc/[l]_{a_0^\dag}_{}="2"
  \ar@/^1pc/[l]^{a_1^\dag}^{}="3"
  \ar@{=>}"2";"3"^{}
  \bullet}),
\end{eqnarray*} i.e. $\Psi_1$ is a morphism with respect to the
horizontal multiplication.

Finally, since the right hand side of \eqref{main map} is a strict
2-group and $(\Psi_1,\Psi_0)$ preserves the horizontal multiplication
strictly,  condition \eqref{condition F2} reduces to
$$
\Psi_1(a_{a_3,a_2,a_1})=\big([(a_3\odot a_2)\odot a_1],1_{H_1}\big).
$$
This holds obviously because $a_{a_3,a_2,a_1} $ being a
reparametrization between $(a_3\odot a_2)\odot a_1$ and
$a_3\odot(a_2\odot a_1)$ must be a homotopy by \cite[Lemma
1.5]{crainic}. Similarly, condition \eqref{condition F2-lr-unit}
holds.

It is clear that $\Psi_0$ sends any smooth path in $P\frkh$ to a
smooth path in $H_0$. Moreover, for a smooth family of homotopies
(parametrized by $u$) $A^u=a^u(t, s) dt+ b^u(t,s)ds$, the solution
$\Delta b^u(t, s)$ of \eqref{eqn:delta b} depends smoothly on $u$.
Thus, both $\Phi_0$ and $\Phi_1$ are smooth.
 \qed

\begin{pro} \label{prop:iso} The natural map \[\varpi: P_2 \frkh /\sim \to H_0 \times H_1 \times_{H_0 \times H_0}
  P_1 \frkh \times P_1 \frkh,   \quad [(a, b, z)] \mapsto \big([a(-,0)],
  [\Delta b], a(-, 0), a(-, 1) \big) \] is an isomorphism.
\end{pro}

We first remark that $P_1 \frkh = P \frkh_0$ and $H_0$ is a quotient
of  $P \frkh$,
thus $\Psi_0:
P_1\frkh \to H_0$ is a surjective submersion of Banach manifolds. Thus
this proposition will automatically imply that $P_2 \frkh/\sim$ is
representible and hence $\tau_2(\int \frkh)$ is a Lie 2-group. The
morphism we demonstrate in last lemma will further be a Lie
2-group morphism.

Now we prove this lemma by constructing an inverse morphism. We notice
that the Lie group $H_1=P^0\frkh_1/\sim$.  Given an
element $\big(h_0,
  h_1, a_0, a_1 \big)$ on the left hand side,  we take a
  representative $\Delta b(s)  \in P^0 \frkh_1$ of $h_1$, then there are
 $a(t, s)$, $\tb(t, s)$ satisfying $\frkh_0$-homotopy equation \eqref{eq:atb} and the
 boundary conditions as in Lemma \ref{lem:homotopy}.  We extend $\Delta b (s)$ to a
  morphism $\Delta b (t, s): [0,1]^{\times 2} \to \frkh_1$ such that
 \begin{equation}\label{eq:extend-db}\Delta b(1, s) = \Delta b(s), \quad \Delta b(0,s)=0, \quad \pat|_{t=0} \Delta
 b(t, s) =0, \quad \pat|_{t=1} \Delta b(t, s) = l_2(a(1, s),
 \Delta(s)).\end{equation}
Such extension always exists. For example, we take
\begin{equation}\label{eq:ext-explicit}
\Delta b(t, s) = \alpha(t) l_2(a(1, s) , \Delta b(s))+
\beta(t) \Delta b(s),
\end{equation} with
$\alpha(0)=\alpha(1)=\beta(0)=\alpha'(0)=\beta'(0)=\beta'(1)=0$, and
$\alpha'(1)=\beta(1)=1$.
 We take \[z(t, s):=l_2(a(t, s), \Delta b(t,
s)) -\pat \Delta b(t, s), \quad b:= \tb - \dM \Delta b. \] Then
\[\pat b - \pae a -l_2(a,b)- \dM z =0,\quad  z(0, s)=z(1,s)=0, \quad
b(1, s)=b(0, s)=\Delta b(0, s)=0,  \] by construction. Thus we may
define a map
\[\zeta: H_0 \times H_1 \times_{H_0 \times H_0}
  P_1 \frkh \times P_1 \frkh \to P_2 \frkh /\sim ,   \quad \big([a_0],
  [\Delta b], a_0, a_1 \big)  \mapsto [(a, b, z)].  \]

\begin{lem} \label{lem:g-well-define}
The map $\zeta$ is well defined.
\end{lem}
\pf If we take another representative $\Delta b^1 \in P^0\frkh_1$ which is
equivalent to $\Delta b$ in $P^0\frkh_1$ via $\Delta b(s, u)$ and $\Delta c(s, u)$,
that is
\begin{equation}\label{eq:dcdb} \pae \Delta c - \pau \Delta b = [\Delta b, \Delta c]_{\frkh_1},
\quad \Delta c(0, u) = \Delta c(1, u) = \Delta b(0, u)= \Delta b(1,
u)=0,
\end{equation} and $a_1 \sim \dM (\Delta b) \odot a_0$ via $a(t,
s)$, $\tb(t, s)$ and $a_1 \sim \dM (\Delta b^1) \odot a_0$ via
$a^1(t, s)$, $\tb^1(t, s)$. Since $\pi_2(H_0)=0$, there is no higher
obstruction between $\frkh_0$-homotopies from being homotopic, so
$(a(t, s)$, $\tb(t, s))$ and  $(a^1(t, s)$, $\tb^1(t, s))$ must be
homotopic via a certain homotopy \[a(t, s, u), \quad \tb(t,s,u),
\quad \tc(t,s,u) \quad \in \frkh_0. \] with boundary conditions
\begin{eqnarray*}a(t, s, 0)&=&a(t, s),\quad a(t, s, 1)=a^1(t, s),\\
\tb(t, s, 0)&=&\tb(t, s),\quad \tb(t, s, 1)=\tb^1(t, s),\quad \tb(0,
s, u)=0,\quad \tb(1, s, u)=\dM (\Delta b(s, u)),\\
\tc (0, s, u) &=&0,\quad \tc(1, s, u)=\dM (\Delta c(s, u)),\quad
\tc(t, 0, u)=0,\quad\tc(t, 1, u)=0.\end{eqnarray*}

Now we repeat the construction of $z(t, s)$ for each $u$, and
we obtain $z(t, s, u)$ and $b(t, s, u)$ with correct boundary
conditions  satisfying \eqref{eq:abz}. We need to show that $(a, b,
z)|_{u=0} \sim (a, b, z)|_{u=1}$.

Firstly, by a similar method, we construct
$y(t, s, u)$ and $c(t, s, u)$ with  correct boundary
conditions\footnote{This amounts to extend $\Delta c (s, u)$ to
  $\Delta c(t, s, u)$ such that $\Delta c(t, s, u)|_{s=0, 1} =0$,
  $\Delta c(0, s, u)=0$, $\Delta c(1, s, u) = \Delta c(s, u)$. The
  boundary condition is a bit different than the case of $\Delta b$,
  however with more information $\Delta c|_{s=0, 1}(s, u)=0$ that
  $\Delta c$ has than $\Delta b$, the same construction of extension works.   } and
satisfying \eqref{eq:cay}. Then $\pau \tb - \pae \tc = [\tc, \tb]_{\frkh_0}$ implies that if we
take
\[ x(t, s, u)= l_2( \Delta b,c) + l_2(b, \Delta c) + \pau \Delta b -
\pae \Delta c +l_2( \dM (\Delta b),\Delta c),  \] we will have
\eqref{eq:cbx}. The boundary condition $x(0, s, u)=0$ is obvious.
The boundary condition $x(1, s, u)=0$ is implied by \eqref{eq:dcdb}. Implied by the boundary condition in \eqref{eq:dcdb}, the extension  $\Delta b(t, s, u)$ from $\Delta
b( s, u)$  according to
\eqref{eq:ext-explicit} for each $u$ satisfies  $\Delta b(t, s,
u)|_{s=0,1}=0$.
This implies the boundary condition $x|_{s=0,1}=0 $.

By straightforward computations, we have
\begin{eqnarray*}
  &&\pau z-\pae y+\pat x\\
  &=&\pau\big(l_2(a,\db)-\pat\db\big)-\pae\big(l_2(a,\dc)-\pat\dc\big)\\
  &&+\pat\big(l_2( \Delta b,c) + l_2(b, \Delta c) + \pau \Delta b -
\pae \Delta c +l_2( \dM (\Delta b),\Delta c)\big)\\
&=&l_2(\pau a,\db)+l_2(a,\pau\db)-\pau\pat\db-l_2(\pae
a,\dc)-l_2(a,\pae\dc)+\pae\pat\dc\\
&&+l_2( \pat\Delta b,c)+l_2( \Delta b,\pat c)+l_2(\pat b, \Delta
c)+l_2(b, \pat\Delta c)+\pat\pau\db-\pat\pae\dc\\
&&+l_2( \pat\dM (\Delta b),\Delta c)+l_2( \dM (\Delta b),\pat\Delta
c)\\
&=&l_2(\pau a-\pat c,\db)+l_2(a,\pau\db-\pae\dc)+l_2(\pat b-\pae
a,\dc)\\
&&+l_2( \pat\Delta b,c)+l_2(b, \pat\Delta c) +l_2( \pat\dM (\Delta
b),\Delta c)+l_2( \dM (\Delta b),\pat\Delta c)\\
&=&-l_2(l_2(a,c)+\dM y,\db)+l_2(a,\pau\db-\pae\dc)+l_2(l_2(a,b)+\dM
z,\dc)\\&&+l_2( \pat\Delta b,c)+l_2(b, \pat\Delta c) +l_2( \pat\dM
(\Delta b),\Delta c)+l_2( \dM (\Delta b),\pat\Delta c),
\end{eqnarray*}
and\begin{eqnarray*}
&&l_2(a,x)-l_2(b,y)+l_2(c,z)\\
&=&l_2(a,l_2( \Delta b,c) + l_2(b, \Delta c) + \pau \Delta b - \pae
\Delta c +l_2( \dM (\Delta b),\Delta c))\\
&&-l_2(b,l_2(a,\dc)-\pat\dc)+l_2(c,l_2(a,\db)-\pat\db).
\end{eqnarray*}
Since $l_2$ satisfies the Jacobi identity, the condition
\eqref{eq:xyz} is equivalent to
$$
-l_2(\dM y,\db)+l_2(\dM z,\dc)+l_2( \pat\dM (\Delta b),\Delta
c)+l_2( \dM (\Delta b),\pat\Delta c)-l_2(a,l_2(\dM\db,\dc))=0.
$$
Compute directly, the left hand side is equal to
\begin{eqnarray*}
 && -l_2(\dM l_2(a,\dc),\db)+l_2(\dM\pat\dc,\db)+l_2(\dM l_2(a,\db),\dc)-l_2(\dM\pat\db,\dc)\\
  &&+l_2( \pat\dM (\Delta b),\Delta
c)+l_2( \dM (\Delta b),\pat\Delta c)-l_2(a,l_2(\dM\db,\dc)),
\end{eqnarray*}
which is equal to zero since $l_2$ satisfies the Jacobi identity.
 Thus, \eqref{eq:xyz} holds.
Therefore, we have $(a, b, z)|_{u=0} \sim (a, b, z)|_{u=1}$ through
$(a, b, c, x, y, z)$. \qed\vspace{3mm}

It is obvious that $\varpi \circ \zeta= id$. To finish the proof of
Theorem \ref{thm:main 3}, we still need to show that $\zeta\circ
\varpi=id$. Given an element $(a, b, z)\in P_2 \frkh$, since
$\varpi$ does not depend on the choice of representative, we choose
a convenient reparametrization such that $z(t, s)|_{s=0,1}=0$. Thus
the solution $\Delta b$ in Lemma \ref{lem:homotopy} also has $\Delta
b(t, s)|_{s=0,1}=0$. Following $\varpi$ then $\zeta$, we first
restrict $\Delta b$ on $t=1$, then extend it again to all $t$ by
\eqref{eq:ext-explicit}, thus we might end up with another $\Delta
b^1$, with the same boundary value, that is when either $t$ or $s$
is $0$ or $1$. Thus Theorem \ref{thm:main 3} follows immediately
from the following lemma:

\begin{lem}\label{lem:extend-db} The map $\zeta$ does not depend on the
  choice of extension with the same boundary value.
\end{lem}
\begin{proof}
We suppose that there are two such  extensions $\Delta b(t, s)$ and
$\Delta b^1(t, s)$. We connect them by $\Delta b^u:= u \Delta b +
(1-u) \Delta b^1 $.  Then  the corresponding $b:= \tb - \dM \Delta
b$ and $b^1:= \tb - \dM \Delta b^1$ are connected by $b^u := \tb -
\Delta b^u$; the corresponding $z$ and $z^1$ are connected by
$z^u:=l_2(a, \Delta b^u) - \pat \Delta b^u $. Now we
 take $a(t, s, u)= a(t, s)$, $b(t, s, u)= b^u$,
$c=0, y=0, x=\pau \Delta b^u$, then it is obviously to see that
\eqref{eq:abz}, \eqref{eq:cay} hold. Equation \eqref{eq:cbx} is implied by the
fact that $\tb$ does not depend on $u$. Equation \eqref{eq:xyz} is
implied by the fact that $a$ does not depend on $u$. The boundary
condition of $x$  is implied by that of $\Delta b$ and $\Delta
b^1$. Thus $(a, b, z)|_{u=0}$ is homotopic to $(a, b, z)|_{u=1}$.
\end{proof}

\section{Application on Integration of (non-strict) Lie 2-algebra morphisms}

Lie's  theorem II tells us that Lie algebra morphisms can integrate
to Lie group morphisms. As pointed out in \cite[Def.
4.2.8]{ScheiberStasheff}  (and also easy to see), an
$L_\infty$-morphism between $L_\infty$-algebras $f:\frkg \to \frkh$
induces a natural map $\int f: \int \frkg \to \int \frkh$ of Kan
complex. Thus applying in the case of Lie 2-algebras, an
$L_\infty$-morphism between Lie 2-algebras  (also called non-strict
Lie 2-algebra morphisms) $f: \frkg \to \frkh$  can integrate to  a
2-group morphism $\tau_2(\int f): \tau_2(\int g) \to \tau_2(\int
h)$. Combining with our result, we have

\begin{cor}\label{thm:main 2} A non-strict Lie 2-algebra morphism $f:
  \frkg\to \frkh$ between two strict Lie 2-algebras integrates to a
  generalized Lie 2-group morphism
\[ (G_0\ltimes G_1 \Rightarrow G_0) \xleftarrow{\sim} \tau_2(\int
\frkg) \xrightarrow{\tau_2(\int f)} \tau_2(\int \frkh)
\xrightarrow{\sim} (H_0\ltimes H_1 \Rightarrow H_0), \]
between the corresponding (simply-connected) Lie group crossed modules. \emptycomment{
 $(\Psi_0, \Psi_1, \Psi_2=id)$:
\begin{equation}\label{main map}\begin{array}{ccc}
P^2(\frkg)&\stackrel{ \Psi_1}{\longrightarrow} &H_0\ltimes H_1\\
\Big\downarrow\Big\downarrow\vcenter{\rlap{ }}&
&\Big\downarrow\Big\downarrow\vcenter{\rlap{
}}\\
P(\frkg)&\stackrel{\Psi_0}{\longrightarrow}  &H_0,
 \end{array}\end{equation}
where $\Psi_0(a(t))=[\mu(a(t))]$, which is the equivalent class of
$\mu(a(t))$ and $\Psi_1$ is given by \eqref{Psi1}. Here the Lie
2-group $\begin{array}{c}
P^2(\frkg)\\
\downarrow\downarrow\vcenter{\rlap{}}\\P(\frkg)
 \end{array}$ is given in Lemma \ref{lem:path2group}.}
\end{cor}
\begin{rmk}
It is fairly easy to integrate a strict morphism which consists of
Lie algebra morphisms $f_i: \frkg_i \to \frkh_i$ preserving all
crossed module structures.   One only needs to integrate $f_i$
individually as a Lie algebra morphism.

The integration of nonstrict morphism is also addressed in the
context of butterflies \cite{noohi:morphism}. Butterflies between
crossed modules are believed\footnote{Private conversation to
Noohi.} to be equivalent to generalized morphisms between strict Lie
2-groups.

Finally, we call the generalized morphism above an
integration of $f$ based on the fact that $\tau_2(\int f)$ should be
considered as  a
canonical integration. However we do not justify the concept of
integration by the inverse procedure, namely differentiation.
\end{rmk}

Now we concentrate on Lie 2-algebra morphisms from a Lie algebra to a
strict Lie 2-algebra. We will see that several interesting objects can
be described by such a morphism, including 2-term representations up to homotopy
of Lie algebras, non-abelian extensions of Lie algebras and up to
homotopy Poisson actions.

We first recall an explicit formulation of $L_\infty$-morphism that we
will mention in the examples:
\begin{defi}
An $L_\infty$-morphism from a Lie algebra $\frkg$ to a strict Lie
2-algebra $L_1\stackrel{\dM}{\longrightarrow} L_0$ consists of
linear maps $\mu:\frkg\longrightarrow L_0$ and
$\nu:\frkg\wedge\frkg\longrightarrow L_1$ such that the obstruction
of $\mu$ being a Lie algebra morphism is given by
\begin{equation}\label{eqn:DGLA morphism c 1}
\mu[X,Y]_\frkg-l_2(\mu(X),\mu(Y))=\dM\nu(X,Y),
\end{equation}
and $\nu$ satisfies the following  condition:
\begin{equation}\label{eqn:DGLA morphism c 2}
l_2(\mu(X),\nu(Y,Z))+c.p.=\nu([X,Y]_\frkg,Z)+c.p.,
\end{equation}
where $c.p.$ means cyclic permutations.
\end{defi}

$\bullet$ {\bf 2-term representations up to homotopy of Lie
algebras}\vspace{2mm}

Associated to any $k$-term complex of vector spaces $\V$, there is a
natural DGLA (differential graded Lie algebra) $\gl(\V)$
\cite{lada-markl,shengzhu2}, which plays the same role as $\gl(V)$
for a vector space $V$ in the classical case. An $L_\infty$-module
\cite{lada-markl} of an $L_\infty$-algebra $L$ is given by an
$L_\infty$-morphism from $L$ to $\gl(\V)$. Associated to any 2-term
complex of vector spaces $\huaV$, by truncation of $\gl(\huaV)$, we
obtain a strict Lie 2-algebra, which we denote by $\End(\huaV)$. The
degree 0 part $\End^0(\huaV)$ is given by
$$
\End^0(\huaV)=\{(A_0,A_1)\in\End(V_0,V_0)\oplus\End(V_1,V_1)|A_0\circ\dM=\dM\circ
A_1\},
$$
and the degree 1 part $\End^1(\huaV)$ is $\Hom(V_0,V_1)$. The Lie
bracket of $\End(\huaV)$ is given by the commutator and the
differential is induced by $\dM$.   It turns out that for 2-term
$L_\infty$-modules of a Lie algebra $\frkg$, it is enough to look at
morphisms to the strict Lie 2-algebra $\End(\huaV)$:

\begin{pro}{\rm\cite{lada-markl}}
A 2-term $L_\infty$-module of a Lie algebra $\frkg$ is given by an
$L_\infty$-morphism from
  $\frkg$ to  $\End(\huaV)$.
\end{pro}

A 2-term $L_\infty$-module of a Lie algebra $\frkg$ is the same as a
representation up to homotopy of the Lie algebra $\frkg$ on a 2-term
complex of vector spaces, see \cite{Camilo rep upto
homotopy,shengzhu1} for more details. Thus Corollary \ref{thm:main
2} can be applied to integrate $L_\infty$-modules $\huaV$ of a Lie
algebra $\frkg$ to that of a Lie group $G$. This is studied further
in \cite{shengzhu2}, where the semidirect product $\frkg \ltimes
\huaV$ is also integrated. It then has application in integrating
omni-Lie algebras and Courant algebroids \cite{shengzhu3}.

\vspace{3mm}

$\bullet $ {\bf Non-abelian extensions of Lie algebras}\vspace{3mm}

It is well known that abelian extensions of a Lie algebra $\frkg$ give rise to a
representation of $\frkg$ and the equivalence classes of extensions are in
one-to-one correspondence with the second cohomology. In the
following, we will see that a non-abelian  extension of a Lie
algebra $\frkg$,  given by a short exact
sequence of Lie algebras
\begin{equation}\label{extension }
 0\stackrel{}{\longrightarrow}\frkk\stackrel{i}{\longrightarrow}
\hat{\frkg}\stackrel{p}{\longrightarrow}
\frkg\stackrel{}{\longrightarrow}0
\end{equation} can be
realized as an $L_\infty$-morphism from Lie algebra $\frkg$ to the
strict Lie 2-algebra
$\frkk\stackrel{\ad}{\longrightarrow}\Der(\frkk)$ (see Example
\ref{ep:derivation}).

By choosing a splitting of $p$,  we can always
assume that $\hat{\frkg}=\frkg\oplus \frkk$ as vector spaces. Then the
Lie bracket $[\cdot,\cdot]_{\hat{\frkg}}$ decomposes as below,
$$
[X_1+k_1,X_2+k_2]_{\hat{\frkg}}=[X_1,X_2]_{\hat{\frkg}}+[X_1,k_2]_{\hat{\frkg}}-[X_2,k_1]_{\hat{\frkg}}+[k_1,k_2]_\frkk,\quad\forall
~X_1+k_1,X_2+k_2\in\frkg\oplus \frkk.
$$
Since $p$ is a morphism of Lie algebras, there is a linear map
$\nu:\frkg\wedge\frkg\longrightarrow\frkk$ such that
$$[X_1,X_2]_{\hat{\frkg}}=[X_1,X_2]_{{\frkg}}+\nu(X_1,X_2).$$
On the other hand, it is straightforward to see that for any
$X\in\frkg$, the action
$[X,\cdot]_{\hat{\frkg}}:\frkk\longrightarrow\frkk$ is a derivation
with respect to the Lie bracket $[\cdot,\cdot]_\frkk$. Thus
$[X,k]_{\hat{\frkg}}=\mu(X)(k)$ for some linear map
$\mu:\frkg\longrightarrow \Der(\frkk)$. One should be very careful
here: $\mu$ is not a Lie algebra morphism!

We rewrite $[\cdot,\cdot]_{\hat{\frkg}}$ as
\begin{equation}\label{eq:bracketext}
[X_1+k_1,X_2+k_2]_{\hat{\frkg}}=[X_1,X_2]_{{\frkg}}+\mu(X_1)(k_2)-\mu(X_2)(k_1)+[k_1,k_2]_\frkk+\nu(X_1,X_2).
\end{equation}
The Jacobi identity of $[\cdot,\cdot]_{\hat{\frkg}}$ gives,
\begin{eqnarray}
\mu([X,Y]_\frkg)(k)-[\mu(X),\mu(Y)]_C(k)&=&[\nu(X,Y),k]_\frkk,\\
~\mu(X)(\nu(Y,Z))+c.p.&=&\nu([X,Y]_\frkg,Z)+c.p..
\end{eqnarray} Moreover, a different splitting $p'$ will give a
homotopic morphism.
We conclude by the following proposition, which we do not claim any
originality (see
\cite{AMR,Nonabeliancohomology, bisheng,stevenson} for various similar
and more general discussions).
\begin{pro}
Given two Lie algebras $\frkg$ and $\frkk$, there is a one-to-one
correspondence between the equivalence classes of non-abelian
extensions of $\frkg$ by $\frkk$ and homotopy classes of
$L_\infty$-morphisms from $\frkg$ to the strict Lie 2-algebra
$\frkk\stackrel{\ad}{\longrightarrow}\Der(\frkk)$.
\end{pro}
Thus apply our result on integration may provide another method to integrate
non-abelian extensions of Lie algebras.


$\bullet ${ \bf Up to homotopy Poisson actions }
\begin{defi}{\rm\cite{Severa}}
An up  to homotopy Poisson action of a Lie algebra $\frkg$ on a
Poisson manifold $(M,\pi)$ is an extension $\frkg_M$ of $\frkg$ by
the Lie algebra $C^\infty(M)$ (with the Poisson bracket
$\{\cdot,\cdot\}_\pi$ used as the Lie bracket), such that for every
$X\in \frkg_M$, the map $C^\infty(M)\longrightarrow
C^\infty(M),~f\longmapsto [X,f]$ is a derivation (i.e. a vector
field).
\end{defi}

Let  $L_M^\pi$ denote the DGLA of multi-vector fields $\Gamma(\wedge
TM)[1]$, with Schouten bracket $[\cdot,\cdot]_S$ and
differential $[\pi,\cdot]$.

As stated in \cite{Severa}, another equivalent formulation of up to
homotopy Poisson action is an $L_\infty$ morphism from $\frkg$ to
the DGLA $L_M^\pi$. We further simplify this  statement. Denote by
$\frkX(M)^\pi$ the set of vector fields preserving the Poisson structure
$\pi$, i.e.
$$
\frkX(M)^\pi=\{X\in\frkX(M),\quad [X,\pi]=0\}.
$$
By truncation, we obtain a strict Lie 2-algebra
$C^\infty(M)\stackrel{[\pi,\cdot]}{\longrightarrow}\frkX(M)^\pi$, of
which the degree $1$ part is $C^\infty(M)$, the degree $0$ part is
$\frkX(M)^\pi$ and the differential is $[\pi,\cdot]$. The extension
 $\frkg_M$ of $\frkg$ by $C^\infty(M)$ is totally determined by a
 linear map $\mu:\frkg\longrightarrow \frkX(M)^\pi$ and a linear map
 $\nu:\frkg\wedge\frkg\longrightarrow C^\infty(M)$, which satisfy
 the  following equation
 \begin{eqnarray*}
\mu([X,Y]_\frkg)-[\mu(X),\mu(Y)]_S&=&[\pi,\nu(X,Y)]_S,\\
\mu(X)(\nu(Y,Z))+c.p.&=&\nu([X,Y]_\frkg,Z)+c.p..
 \end{eqnarray*}
Thus, we have
\begin{pro}
There is a one-to-one correspondence between up to homotopy Poisson
actions of Lie algebra $\frkg$ on Poisson manifolds $(M,\pi)$ and
$L_\infty$-morphisms $(\mu,\nu)$ from $\frkg$ to the strict Lie
2-algebra
$C^\infty(M)\stackrel{[\pi,\cdot]}{\longrightarrow}\frkX(M)^\pi$.
\end{pro}
\begin{rmk} We only need to use the fact $\pi_2(H_0)=0$ in the
  construction of $g$ in the last section. Without this condition, we
  will still have a morphism even though not a Morita morphism. The space of $\frkX(M)^\pi$
  is infinite dimensional and does not admit a Banach
  structure. However, there
  is also certain infinite-dimensional calculus available in this
  case (see for example \cite[App.A]{wz:int}). Thus our result can not
 be applied directly, however certain modification may be applied.
\end{rmk}

\end{document}